\documentclass[amsmath,amssymb,prl,twocolumn,final]{revtex4-2}
\usepackage{dcolumn}
\pdfoutput=1
\usepackage[english]{babel}
\usepackage[utf8]{inputenc}
\usepackage[T1]{fontenc}
\usepackage{amsthm}
\usepackage{amsfonts}
\usepackage{url}
\usepackage{tensor}
\usepackage[colorlinks = true,
            linkcolor = blue, 
            linkcolor = blue,
            urlcolor  = blue,
            citecolor = blue,
            anchorcolor = blue]{hyperref}
\usepackage{graphicx}
\usepackage{mathrsfs}
\usepackage{enumerate}
\usepackage{caption}
\usepackage{bm}
\usepackage{autonum}
\usepackage{orcidlink}

\usepackage{xcolor}

\usepackage{amsthm}

\newtheorem{lemma}{Lemma}     

\newcounter{mnotecount}[section]

\newcommand{\beq}{\begin{eqnarray}}
\newcommand{\eeq}{\end{eqnarray}}
\newcommand{\ben}{\begin{eqnarray*}}
\newcommand{\een}{\end{eqnarray*}}

\newtheorem{theorem}{Theorem}
\newtheorem*{theorem*}{Theorem}
\theoremstyle{definition}

\newcommand\ringring[1]{%
  {
   \mathop{\kern0pt #1}\limegre{
     \vbox to-1.85ex{
       \kern-2ex 
       \hbox to 0pt{\hss\normalfont\kern.1em \r{}\kern-.45em \r{}\hss}
       \vss 
     }   }
  }
}

\begin{document}
\title{Some inequalities among curvature invariants}
\author{Sebastian J. Szybka\, \orcidlink{0000-0003-3648-9285}}
\affiliation{Astronomical Observatory, Jagiellonian University}
\author{Yaroslava Kravetska\, \orcidlink{0009-0006-6666-7360}}
\affiliation{Astronomical Observatory, Jagiellonian University}
\author{Kornelia Nikiel\, \orcidlink{0009-0002-1834-4904}}
\affiliation{Institute of Theoretical Physics, Jagiellonian University}
\begin{abstract}
We prove an infinite sequence of inequalities among scalar polynomial invariants of symmetric rank-2 tensors of Segre types $A1$, $A3$, and $B$. In particular, these inequalities apply to the Ricci tensor and the energy-momentum tensor. If at least one of them is violated by the Ricci tensor, then the Einstein equations force violation of all classical energy conditions. In addition, we use one of the inequalities to generalize the known relation between the second Ricci invariant and the Kretschmann scalar.
\end{abstract}
\maketitle{}

\section{Introduction}

In general relativity and other metric gravity theories, curvature invariants provide coordinate independent characterization of the geometrical and physical properties of spacetime \cite{exact2003}. They can be applied in the equivalence problem and spacetime classification, to detect singularities and black hole horizons, and in many other contexts \cite{MacCallum:2015zaa}. Given their invariance, it is natural to use them to define physical measures in spacetime (e.g., gravitational entropy \cite{Rudjord_2008}, \cite{Clifton_2013}). Algebraic relations between curvature invariants, called syzygies, appear in spacetimes of given dimension and/or within particular classes of spacetimes \cite{zakhary}. In this paper, we address a problem related to syzygies---we investigate inequalities between polynomial curvature invariants.

This problem, in full generality, seems to be difficult and largely unexplored---a complete treatment lies beyond the scope of this paper. We restrict our attention to simple cases. Specifically, in a recent paper \cite{WCI} some inequalities between curvature invariants were proved for certain classes of spacetimes. We show that this problem can be tackled more generally using the Segre classification. The discovered properties of curvature invariants are algebraic consequences of the canonical form of generic symmetric rank-2 tensors of specific Segre types.

The inequalities between the Ricci curvature invariants that follow from the theorem proved in this paper translate, via the Einstein equations, into constraints on the energy-momentum tensor, satisfied by all physically relevant fields. Therefore, violation of these inequalities can be used to rule out unphysical spacetime geometries and shed light on the construction of exact solutions of the Einstein equations via the Synge method \cite{Ellis:2023css}.

\section{Segre classification}

Let $\bm P$ be the symmetric rank-2 tensor on the Lorentzian $D$-dimensional manifold $(M,g)$, where $D\geq 2$. At any point $p\in M$, the metric $\bm g$ identifies $\bm P$ with a linear map $\bm{\hat P}:T_pM\to T_pM$ represented by $P^a_{\; b}$. The eigenvalue problem for this map leads to the Segre classification and allows one to cast $\bm P$ into a canonical form.

The canonical form of $\bm P$ for a given Segre class in $D$ dimensions is a straightforward generalization of the well-known four-dimensional case \cite{exact2003}. It reads as follows
\begin{equation}\label{canonic}
\begin{split}
A1:\quad P_{ab} &= \sum_{i=1}^{D-1}\lambda_i \,x^{(i)}_a x^{(i)}_b - \lambda_D u_a u_b\;,\\
	A2:\quad P_{ab} &= \sum_{i=1}^{D-2}\lambda_i \,x^{(i)}_a x^{(i)}_b + \lambda_{D-1} k_{(a}l_{b)}+ \lambda_D (k_a k_b - l_a l_b)\;,\\
	A3:\quad P_{ab} &= \sum_{i=1}^{D-2}\lambda_i \,x^{(i)}_a x^{(i)}_b -2 \lambda_{D-1} k_{(a}l_{b)}\pm k_a k_b\;,\\
	B: \quad P_{ab} &= \sum_{i=1}^{D-2}\lambda_i \,x^{(i)}_a x^{(i)}_b  + \lambda_{D-1} k_{(a}l_{b)}+ k_{(a} x_{b)}\;,\\
\end{split}
\end{equation}
where $\lambda_i\in\mathbb{R}$ are eigenvalues or linear combinations thereof, $(\bm{x^{(1)}},\dots , \bm{x^{(D-1)}}, \bm u)$ is an orthonormal tetrad with $\bm u$ timelike, and $(\bm{x^{(1)}},\dots , \bm{x^{(D-2)}}, \bm k, \bm l)$ is a real partially null tetrad. The partially null tetrad satisfies
\begin{equation}
	\bm{x^{(i)}}\cdot \bm{x^{(k)}}=\delta^{ik}\;,\quad \bm k \cdot \bm l=-1\;,\quad \bm k^2=\bm l^2=0\;,
\end{equation}
and the remaining scalar products vanish.

We say that the spacetime (or a region thereof) is of Segre type $A1$, $A2$, $A3$, or $B$, if the Ricci tensor $\bm R$ is of the corresponding Segre type. The canonical tetrad associated with the Ricci tensor is called the Ricci principal tetrad.

The Segre types split further into subclasses depending on the degeneracies of eigenvalues and types of eigenvectors---timelike, spacelike, or null \cite{exact2003}. For example the type $A2$ contains two subclasses denoted in the Segre notation as $[11,Z\bar{Z}]$, $[(11),Z\bar{Z}]$. The first subclass has two distinct eigenvalues associated with spacelike eigenvectors (denoted with $11$) and a pair of complex conjugate eigenvalues associated with complex eigenvectors (denoted with $Z\bar{Z}$). A comma separates eigenvalues corresponding to spatial eigenvectors from others. The round brackets in the symbol $[(11),Z\bar{Z}]$ representing the second subclass indicate that the there are two eigenvectors (spatial) associated with the same eigenvalue (algebraic multiplicity $2$). The types of eigenvectors in the eigenspace are more explicitly indicated in the Plebański notation \cite{Pleb}. The both subclasses mentioned above are represented by $[S_1-S_2-Z-\bar{Z}]_{(1111)}$ and $[2S-Z-\bar{Z}]_{(111)}$, respectively. The orders of the corresponding factors in the minimal polynomial (nilpotency indices) are indicated in the subscript. They correspond to sizes of Jordan blocks.

\begin{lemma}\label{lem1}
 Let $(M,g)$ be a spacetime, and let a rank-2 tensor $\bm{ P'}$ be related to a symmetric rank-2 tensor \mbox{$\bm P$} at \mbox{$p\in M$} by $$\bm{P'}=\alpha \bm P+\beta \bm g\;,$$ with $\alpha,\beta\in \mathbb{R}$ and $\alpha\neq 0$. Then $\bm{P'}$ is of the same Segre type at $p$ as $\bm P$.
\end{lemma}
\begin{proof}
	The map $\bm{\hat P'}:T_pM\to T_pM$ associated with $\bm{P'}$ is given by an affine transformation $\bm{\hat P'}=\alpha \bm{\hat P}+\beta\,\bm{id}$, where $\bm{id}$ is the identity. If $\bm x$ is an eigenvector of $\bm{\hat P}$ with eigenvalue $\lambda$, then $\bm{\,x}$ is an eigenvector of $\bm{\hat P'}$ with eigenvalue $\alpha \lambda+\beta$. The Jordan normal form of $P'^a_{\;\;b}$ (representing $\bm{\hat P'}$) has blocks of the same sizes as that of $P^a_{\;b}$ (representing $\bm{\hat P}$), hence the Segre type is preserved.
\end{proof}

The lemma above implies that the Segre type of the traceless Ricci tensor $\bm S$ ($S_{ab}=R_{ab}-R/D\,g_{ab}$) is the same as that of the Ricci tensor $\bm R$. If the Einstein equations hold, then the Segre type of the energy-momentum tensor $\bm T$ is also the same as that of the Ricci tensor $\bm R$. The cosmological constant does not affect the Segre type. In summary, $\bm R$, $\bm S$, and $\bm T$ share the same canonical form presented in the previous section. The canonical orthonormal or partially null tetrad is common to all three tensors, but the coefficients $\lambda_i$ differ by linear transformations. In metric theories of gravity other than general relativity, the Segre type of the Ricci tensor $\bm R$ and the energy-momentum tensor $\bm T$ typically differ.

\section{Invariants}

Let $P_n$ denote the trace of the $n$th power of $\bm P$, i.e., $$P_n = P_{a_1 b_1} P_{a_2 b_2} \cdots P_{a_n b_n} g^{b_1 a_2} g^{b_2 a_3} \cdots g^{b_{n-1} a_n} g^{b_n a_1}\,.$$ We have
\begin{equation}
	P_1 = P, \quad P_2 = P_{ab} P^{ab}, \quad P_3 = P_{ab} P^{bc} P_c{}^a, \quad \dots
\end{equation}
If $P$ is the Ricci tensor $R$, then $R_1=R$ is the Ricci scalar and $R_n$ is the $n$th Ricci invariant.
\begin{theorem}\label{thm1}
	Let $(M,g)$ be a $D$-dimensional spacetime, $D\geq 2$, and $\bm{P}$ be a rank-2 symmetric tensor of Segre type $A1$, $A3$ or $B$. Then 
	\begin{equation}
		P_s^{2m}\leq D^{2m-s} P_{2m}^s\;,
	\end{equation}
for every $s,m\in \mathbb{N}$ such that $1\leq s<2m$.
\end{theorem}

\medskip
\noindent\textit{Comments.} The inequality above actually represents an infinite sequence of inequalities. It was proved in the article \cite{WCI} for $\bm R$, $\bm S$, and the traceless part of the energy-momentum tensor $\bm T$ (which follows trivially from the Einstein equations) in static spherically symmetric spacetimes. These spacetimes are of Segre type $A1$ (the Segre type of $\bm R$) and it follows from Lemma \ref{lem1} that $\bm S$, $\bm T$ are also of Segre type $A1$, so the results of the article \cite{WCI} are a special case covered by Theorem \ref{thm1}.

	The Cayley-Hamilton theorem implies that, in $D$ dimensions, one can use the characteristic polynomial to express the higher invariants $P_{n\geq D}$ in terms of lower order invariants $P_1,\dots,P_{D-1}$. However, for \mbox{$m < D/2$}---the case of interest for our inequality---the Cayley-Hamilton theorem provides no additional information. Thus, the inequality above does not follow directly from it. In the next section we present an example of the Schmidt spacetime in which, in the $A2$ Segre regions, the inequality is violated for the Ricci tensor $\bm{R}$ and several choices of $s$ and $m$. In other words, it does not hold trivially.

\begin{proof}

Consider the generalized mean inequality in the following form \cite{WCI}. It holds for real numbers $x_i$ ($i=1,\ldots,D$) and $s,m\in \mathbb{N}$ with $s<2m$
\begin{equation}\label{gamqm}
	\left(\sum_{i=1}^D x_i^s \right)^{2m} \leq D^{2m-s}\left(\sum_{i=1}^D x_i^{2m}\right)^s\;.
\end{equation}
	The even power $2m$ ensures that the signs of $x_i$ do not matter on the right-hand side of the inequality (the inequality is originally stated for $x_i\geq 0$).

	For the $A1$ class we have
	\begin{equation}\label{IN1}
	P_s=\sum_{i=1}^D\lambda_i^s\;.
\end{equation}
Hence $P_s^{2m}\leq D^{2m-s} P^s_{2m}$ is equivalent to the inequality~\eqref{gamqm}.

	For the class $A3$, with $D-1$ functions $\lambda_i$, we have 
\begin{equation}\label{IN2}
	    P_s=\sum_{i=1}^{D-1}\lambda_i^s+\lambda_{D-1}^s=\sum_{i=1}^D\lambda_i^s\;,
\end{equation}
where we have introduced the additional symbol \mbox{$\lambda_D=\lambda_{D-1}$.} Therefore, $P_s^{2m}\leq D^{2m-s} P^s_{2m}$ is again equivalent to the inequality \eqref{gamqm}.

For the Segre class $B$, with $D-1$ functions $\lambda_i$, we find
\begin{equation}
	\begin{split}\label{IN3}
		P_s=&\sum_{i=1}^{D-2}\lambda_i^s+\frac{(-1)^s}{2^{s-1}}\lambda_{D-1}^s\\=&\sum_{i=1}^{D-2}\lambda_i^s+\left(\frac{-1}{2}\right)^s\lambda_{D-1}^s+\left(\frac{-1}{2}\right)^s\lambda_{D-1}^s\;.
    \end{split}
\end{equation}
Set $\lambda'_{D-1}=\left(\frac{-1}{2}\right)\lambda_{D-1}$, $\lambda'_D=\left(\frac{-1}{2}\right)\lambda_{D-1}$, and $\lambda'_i=\lambda_i$ for $i=1,\dots,D-2$. Then
\begin{equation}
	    P_s=\sum_{i=1}^{D}\lambda_i'^s\;,
\end{equation}
and $P_s^{2m}\leq D^{2m-s} P^s_{2m}$ is again equivalent to the inequality \eqref{gamqm}.

\end{proof}

\begin{theorem}\label{thm2}
	Let $(M,g)$ be a $D$-dimensional spacetime, $D\geq 2$, of Segre type $A1$, $A3$ or $B$. Let $I_1$ denote the square of the Weyl tensor and let $K$ be a Kretschmann scalar. If $I_1\geq 0$, then
	\begin{equation}
		2 R_2\leq (D-1)K\;.
	\end{equation}
\end{theorem}
\begin{proof}
	The theorem above follows directly from \mbox{Theorem \ref{thm1}} applied to the Ricci tensor and the inequality presented in the article \cite{WCI}.

	For $D=2$ we have $K=2R_2=R^2_1$, so $2 R_2\leq (D-1)K$ reduces to $R_1^2\leq R_1^2$, which is trivially true.

	For $D\geq 3$ the Kretschmann scalar can be written in terms of the Ricci invariants $R_1$, $R_2$ and the invariant $I_1=C_{abcd}C^{abcd}$, where $C_{abcd}$ is the Weyl tensor
\begin{equation}
	K=I_1+\frac{4}{D-2}R_2-\frac{2}{(D-1)(D-2)}R_1^2\;.
\end{equation}
Rewriting $2 R_2\leq (D-1)K$ gives
\begin{equation}
R_1^2\leq DR_2+\frac{1}{2}(D-2)(D-1)I_1\;.
\end{equation}
Since by our assumptions $I_1\geq 0$, it suffices to use Theorem \ref{thm1} with $s=m=1$ (it yields $R_1^2\leq D R_2$) to complete the proof.
\end{proof}

The Weyl tensor vanishes for $D=3$. It also vanishes for $D=4$ in spacetimes of Petrov type O; hence $I_1=0$ in both cases. In $D=4$ and for Petrov types N, III the Weyl tensor is nonzero but still satisfies $I_1=0$. In all these cases, Theorem \ref{thm1} with $s=m=1$ and Theorem \ref{thm2} reduce to the same inequality.

\section{Examples and discussion}

We point out that most classical fields found in nature yield, via the Einstein equations, a Ricci tensor of Segre type $A1$, so Theorems \ref{thm1}, \ref{thm2} apply to a wide class of physically relevant spacetimes \cite{Mart_n_Moruno_2021}.

The special case $s=1$, $m=1$ of the inequality from Theorem \ref{thm1}, namely $R_1^2\leq D R_2$, together with the inequality from Theorem \ref{thm2} were proved for the Ricci tensor in the article \cite{WCI} for static spacetimes, generalized Robertson-Walker spacetimes, and generalized Bianchi I spacetimes. All of these are particular members of class $A1$. The inequality from Theorem \ref{thm1} was also proved in the article \cite{WCI} for $\bm R$, $\bm S$, and $\bm T$ in static spherically symmetric spacetimes.

Lorentzian Einstein spaces are of type $A1$, and it is immediate that Theorems \ref{thm1}, \ref{thm2} hold in them.

In Segre class $A3$ there are also spacetimes of physical interest, such as those sourced by a null-dust energy-momentum tensor $T_{ab}=\rho(x)k_a k_b$, where $\bm k\cdot \bm k=0$. An important example is the Vaidya spacetime \cite{vaidya}. 

Using the Einstein equations, for this traceless null-dust tensor we have 
$$R_{ab}=8\pi(T_{ab}-\frac{1}{2}Tg_{ab})=8\pi T_{ab}=8\pi\rho(x)k_a k_b\;.$$ Hence, $R_s=0$ for $s\in \mathbb{N}$, and Theorem \ref{thm1} holds trivially.
The Ricci tensor $R^a_{\;b}$ has one eigenvalue $0$ of algebraic multiplicity $4$. There exists one null and two spacelike eigenvectors associated with this eigenvalue. The generalized eigenspace (primary subspace) corresponding to the null vector is two-dimensional---the nilpotency of $R^a_{\;b}$ is $2$, hence the size of the corresponding Jordan block is also $2$. Such spacetimes are of Segre type $A3$ $[(11,2)]$ ($[4N]_{(2)}$ in Plebański notation). 

Next, we restrict our attention to the Vaidya spacetime. In standard coordinates we have $$K=I_1=48m(u)^2/r^6>0\;,$$ where $u$ is a null coordinate, $r$ is the areal radius and $m(u)$ is the mass function. Since $I_1>0$, Theorem \ref{thm2} also holds for the Vaidya spacetime.

The unphysical Schmidt metric \cite{Schmidt_2003} considered in the paper \cite{WCI} has a form
\begin{equation}
    ds^2=-dt^2+2\,y\,z\,dt dx+dx^2+dy^2+dz^2\;.
\end{equation}
It is possible to solve the Ricci eigenvalue problem for this metric directly, but resulting formulas are too long to be usefully quoted here. One can show that the Ricci tensor is of Segre type $A1$ only for $y\,z=0$. Namely, for $y=z=0$ it is of Segre type $A1$ $[(111,1)]$ ($[4T]_{(1)}$ in Plebański notation). For $y\,z=0$ but $y\neq0$ or $z\neq 0$ it has Segre type $A1$ $[1(11,1)]$ ($[S-3T]_{(11)}$  in Plebański notation). We verified that Theorem \ref{thm1} holds in the $A1$ region. In the remaining region of the spacetime the Ricci tensor is of Segre type $A2$, and our theorems do not apply. In particular, it has type $[11,Z\bar{Z}]$ ($[S_1-S_2-Z-\bar{Z}]_{(1111)}$ in Plebański notation). The second possibility within the $A2$ class (a spatial eigenvalue of algebraic multiplicity $2$) corresponds to vanishing of the $12$th order polynomial. One can show that it is equivalent to vanishing of an explicitly nonnegative polynomial that is equal to zero only when $y=z=0$
\begin{equation}\label{poly}
	(y^4+z^4)\,(z^2\,y^2-1)^2+2\,y^2\,z^2\,(y^4\,z^4+14\,y^2\,z^2+17)=0\;.
\end{equation}

\begin{figure}[t!]
    \centering
    \includegraphics[width=0.45\textwidth]{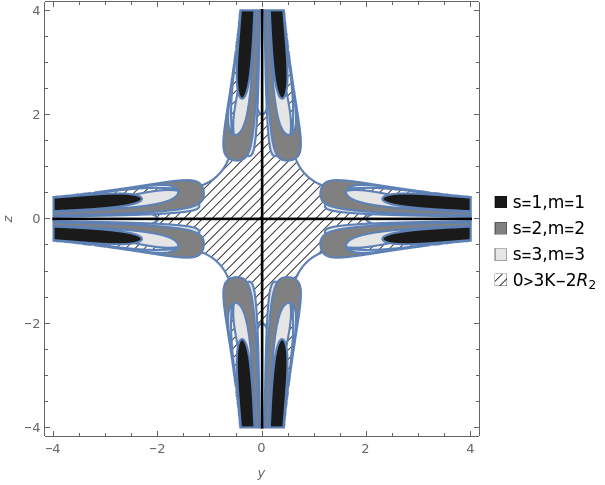} 
	\caption{Regions where the inequalities $R_s^{2m}\leq D^{2m-s} R_{2m}^s$ and $2 R_2\leq (D-1)K$ are violated for the Schmidt spacetime. The bold black lines correspond to $y\,z=0$, where the Ricci tensor is of Segre type $A1$ and Theorem \ref{thm1} holds.}
    \label{fig1}
\end{figure}

Figure \ref{fig1} shows the regions in which the inequalities $R_s^{2m}\leq D^{2m-s} R_{2m}^s$ and $2 R_2\leq (D-1)K$ are violated. These regions intersect, but neither contains the other. We also note that the inequalities stated in Theorems \ref{thm1} and \ref{thm2} can hold in some regions where the Ricci tensor is of Segre type $A2$, which are not covered by our theorems.

In special cases, Theorem \ref{thm1} leads to other nontrivial inequalities. Ricci invariants are often defined in terms of the traceless Ricci tensor $\bm S$ rather than $\bm R$. Consider the class of nonvacuum spacetimes with vanishing \mbox{$S_3=S^{a}_{\;b}S^b_{\;c}S^c_{\;a}$,} studied in the papers \cite{Lake:2019rwe}, \cite{Bisson:2023sxn}. (Note that invariant indices are often labeled differently from ours, in particular, $r_2=S_3$ is the traceless analogue of $R_3$.) We have
\begin{equation}
S_3=0\iff D^2\,R_3-3D\,R_1R_2+2\,R_1^3=0\;.
\end{equation}
The spacetimes analyzed in the papers \cite{Lake:2019rwe}, \cite{Bisson:2023sxn} are static and, as such, belong to Segre class $A1$. We note that the much-misinterpreted solution to the Einstein equations, the so-called Kiselev black hole \cite{Kiselev:2002dx}, belongs to this static $S_3=0$ class. 

Using Theorem \ref{thm1} with $s=1$, $m=1$ and assuming $R_1\neq 0$, we obtain the nontrivial inequality
\begin{equation}
	D^2 \frac{R_3}{R_1^3}\geq 1\;.
\end{equation}
If $R_1=0$, then $S_3=0$ implies $R_3=0$.

\section{Conclusions}

Any symmetric rank-2 tensor admits a simple canonical form which follows from the Segre classification. This form is particularly useful for studying inequalities between curvature invariants. Inequalities such as those in the paper \cite{WCI} can be extended to spacetimes of specific Segre types and to other symmetric rank-2 tensors. This approach is more efficient than working directly with the metric. 

We have proved two theorems. Theorem \ref{thm1} establishes an infinite family of inequalities among invariants constructed out of rank-2 symmetric tensors, in particular among the Ricci invariants. These inequalities hold for tensors of Segre types $A1$, $A3$, and $B$, but are not necessarily satisfied for tensors of type $A2$. Theorem \ref{thm2} relates the second Ricci invariant to the Kretschmann scalar and follows naturally from Theorem \ref{thm1}. Both theorems generalize the results of the paper \cite{WCI}. 

The inequalities among traces of $n$th powers of the energy-momentum tensor do not impose additional constraints on physical quantities (such as the energy density or pressure of a perfect fluid). However, they do constrain, via the Einstein equations, the class of physically relevant spacetime geometries. 

It follows from Lemma \ref{lem1} that, in general relativity, the Segre type of the Ricci tensor is the same as that of the energy-momentum tensor. If at least one of the inequalities from Theorem \ref{thm1} fails for the Ricci tensor, then the energy-momentum tensor is necessarily of Segre type $A2$. Type $A2$ corresponds to type $IV$ in the Hawking-Ellis classification \cite{Hawking:1973uf}. In the Segre classification, it is singled out as the only class with no timelike or null eigenvector. As pointed out by Hawking and Ellis \cite{Hawking:1973uf}, there are no observed fields whose energy-momentum tensors have this form. In fact, all classical energy conditions fail for this type of the energy-momentum tensor \cite{ec2017}. Therefore, if the Ricci tensor violates at least one inequality from Theorem \ref{thm1}, then the spacetime is unphysical. The reverse implication does not hold. The Schmidt metric presented in this paper shows that there exist unphysical spacetimes in which at least some of the inequalities from Theorem \ref{thm1} hold. 

The restriction on physical geometries which follows from Theorem \ref{thm1} has interesting consequences for the singularity formation in general relativity. The blow-up of the trace of the $s$th power of the Ricci tensor $R_s$, must be accompanied by the blow-up of the trace of the $2m$th power of the Ricci tensor $R_{2m}$, for every $1\leq s<2m$.

Theorem \ref{thm2} shows that, under additional assumptions, the inequalities from Theorem \ref{thm1} impose constraints on other invariants, such as the Kretschmann scalar. This fact has been recently used to derive constraints on regular black holes supported by nonlinear electromagnetic fields \cite{bhnle}.

Finally, we point out, that in the vast majority of metric theories of gravity other than general relativity, the Segre types of the Ricci tensor and the energy-momentum tensor differ. The algebraic restrictions on physical geometries derived in this paper therefore distinguish general relativity from other theories.

\section*{Acknowledgments}

Some calculations were performed in Wolfram Mathematica using the xAct package \cite{xAct}. AI language tools (OpenAI GPT, Microsoft Copilot, and Google Gemini) were used for language polishing and for suggesting an elegant form for the polynomial~\eqref{poly}, which we verified independently.

\section*{Data availability}

No data were created or analyzed in this study.

\bibliographystyle{apsrev4-1}
\setcitestyle{authortitle}
\bibliography{report}

\end{document}